\documentclass[10pt]{article}

\usepackage{a4wide}

\usepackage{amsmath}
\usepackage{amsthm}
\usepackage{amssymb}
\usepackage{graphicx}
\usepackage{pstricks}

\newcommand{\probtitle}[1]{\textsc{#1}}

\newcommand{\qedfill}[0]{ }

\newtheorem{remark}{Remark}
\newtheorem{lemma}{Lemma}
\newtheorem{theorem}{Theorem}
\newtheorem{proposition}{Proposition}
\begin{document}

%%%%%% FREE
\newcommand\thetitle{The firefighter problem with more than one firefighter on trees}

\title{\textbf{\thetitle}}
\author{Cristina Bazgan \and Morgan Chopin \and Bernard Ries\\
 {\small Universit\'{e} Paris-Dauphine, LAMSADE,}\\
  {\small Place du Marchal de Lattre de Tassigny, 75775 Paris Cedex 16, France.}\\
  {\small \{bazgan, chopin, ries\}@lamsade.dauphine.fr }\\
 }

\date{ }

\maketitle
%%%%%% FREE

%%%%%% ELSEVIER
%\title{The firefighter problem with more than one firefighter on trees}

%\author[cb]{Cristina Bazgan}
%\ead{bazgan@lamsade.dauphine.fr}

%\author[mc]{Morgan Chopin}
%\ead{chopin@lamsade.dauphine.fr}

%\author[br]{Bernard Ries}
%\ead{ries@lamsade.dauphine.fr}

%\address[cb,mc,br]{Universit\'e Paris-Dauphine, LAMSADE, France}
%%%%%% ELSEVIER

\begin{abstract}
In this paper we study the complexity of the firefighter
problem and related problems on trees when more than one firefighter is available at each time
step, and answer several open questions of \cite{finbow2009}.
More precisely, when $b \geq 2$ firefighters are allowed at each time step, the problem is NP-complete
for trees of maximum degree $b+2$ and polynomial-time solvable for trees of
maximum degree $b+2$ when the fire breaks out at a vertex of
degree at most $b+1$. Moreover we present a polynomial-time algorithm for a subclass of trees, namely $k$-caterpillars.
\end{abstract}

%\begin{keyword}
%Firefighter problem \sep NP-hard \sep trees \sep $k$-caterpillars \sep polynomial-time
%\end{keyword}

%\end{frontmatter}

\section{Introduction}

Modeling a spreading process in a network is a widely studied topic and often relies on a graph theoretical approach (see \cite{chen2008,dreyer2009,finbow2009,kempe2003,ng2008,scott2006}). %(see [5,8,12,13,16,17]). 
Such processes occur for instance in epidemiology and social sciences. Indeed, the spreading process could be the spread of an infectious disease in a population 
or the spread of opinions through a social network. Different objectives may then be of interest, for instance minimizing the total number of 
infected persons by vaccinating at each time step some particular individuals, or making sure that some specific subset of individuals does not get infected at all, etc...

The spreading process may also represent the spread of a fire. The associated firefighter problem, introduced in \cite{hartnell1995}, 
has been studied intensively in the literature (see for instance \cite{anshelevich2009,cai2008,develin2007,finbow2007,finbow2009,hartnell1995,hartnell2000,IKM11,king2010,macgillivray2003,ng2008}). In this paper, 
we consider some generalizations and variants of this problem which is defined as follows.
%Modeling a diffusion process in a social network is a widely studied problem and relies often on a graph theoretical approach \cite{finbow2009,dreyer2009,ng2008,scott2006,chen2008,kempe2003}. The diffusion process could be the spread of a virus, a fire, or even an idea.  
%In such problem we are mainly concerned with the following two questions: can we find a subset of the most influential individuals? \cite{dreyer2009,chen2008,kempe2003} and, conversely, can we contain an infection and/or minimize the number of infected persons? \cite{chalermsook2010,develin2007,hartnell1995}. In this paper, we consider the second question, and more precisely, we study the firefighter problem. 
%The firefighter problem was introduced in \cite{hartnell1995}. It is a dynamic problem defined as follows.
Initially, a fire breaks out at some special
vertex $s$ of a graph. At each time step, we have to choose one
vertex which will be protected by a firefighter. Then the fire
spreads to all unprotected neighbors of the vertices on fire. The
process ends when the fire can no longer spread, and then all
vertices that are not on fire are considered as saved. The
objective consists of choosing, at each time step, a vertex which
will be protected by a firefighter such that a maximum number of
vertices in the graph is saved at the end of the process. 

The firefighter problem
was proved to be NP-hard for bipartite graphs
\cite{macgillivray2003}. Much stronger results were proved later
\cite{finbow2007} implying a dichotomy: the firefighter problem is
NP-hard even for trees of maximum degree three and it is solvable
in polynomial-time for graphs with maximum degree three, provided
that the fire breaks out at a vertex of degree at most two.
Moreover, the firefighter problem is NP-hard for cubic graphs
\cite{king2010}. From the approximation point of view, the
firefighter problem is $\frac{e}{e-1}$-approximable on trees \cite{cai2008} and
it is not $n^{1-\varepsilon}$-approximable on general graphs for any $\epsilon \in (0,1)$
\cite{anshelevich2009}, if P$\neq$ NP. Moreover for trees where vertices have at most three children, 
the firefighter problem is $1.3997$-approximable \cite{IKM11}. Finally, the firefighter problem is polynomial-time solvable for caterpillars and P-trees \cite{macgillivray2003}.

A problem related to the firefighter
problem, denoted by {\sc $S$-Fire}, was introduced in
\cite{king2010}. It consists of deciding if there is a strategy
of choosing a vertex to be protected at each time step such that  all vertices
of a given set $S$ are saved. {\sc $S$-Fire} was proved to be
NP-complete for trees of maximum degree three in which every leaf is
at the same distance from the vertex where the fire starts and $S$ is the set of leaves.

In this paper, we consider a generalized version of {\sc
$S$-Fire}. We denote by {\sc $b$-Save}, where $b \geq 1$ is an integer,
the problem which consists of deciding if we can choose at most $b$ vertices
to be protected by firefighters at each time step and save all the vertices
from a given set $S$. Thus, {\sc $S$-Fire} is equivalent to {\sc $1$-Save}.
The optimization version of {\sc $b$-Save} will be denoted by {\sc Max $b$-Save}.
This problem consists of choosing at most $b$ vertices to be protected at each time
step and saving as many vertices as possible from a given set $S$. Hence, {\sc Max $1$-Save}
corresponds to the firefighter problem when $S$ is the set of all vertices of the graph.
{\sc Max $b$-Save} is known to be 2-approximable for trees  when $S$ is the
set of all vertices \cite{hartnell2000}. 

A survey on the firefighter problem and related problems can
be found in \cite{finbow2009}. In this survey, the authors presented a list of
open problems. Here, we will answer three of these open questions (questions 2, 4, and 8).

The first question asks for finding algorithms and complexity results of {\sc $b$-Save}
when $b\geq 2$. We show that for any fixed $b\geq 2$, {\sc
$b$-Save} is NP-complete for trees of maximum degree $b+2$ when $S$ is the set of all leaves. Moreover, we show that for any fixed $b \geq 2$,
{\sc Max $b$-Save} is NP-hard for trees of maximum degree $b+3$ when $S$ is the set of all vertices.
Finally, we show that for any $b \geq 1$, {\sc $b$-Save} is polynomial-time solvable for trees of
maximum degree $b+2$ when the fire breaks out at a vertex of degree at most $b+1$.

The second question asks if there exists a constant $c > 1$ such that the greedy strategy of protecting, at each time step, a vertex of highest degree adjacent to a burning vertex gives a polynomial-time $c$-approximation for the firefighter problem on trees. We give a negative answer to this question.

Finally, the third question asks for finding classes of trees for which the firefighter problem can be solved in
polynomial time. We present a polynomial-time algorithm to solve {\sc Max $b$-Save}, $b \geq 1$, in $k$-caterpillars a subclass of trees.

Our paper is organized as follows.
Definitions, terminology and preliminaries are given in  Section~\ref{s:prelim}.
In Section~\ref{s:trees} we establish a dichotomy on the complexity of {\sc $b$-Save} and show that the greedy strategy mentioned above gives no approximation guarantee. In Section~\ref{s:caterpillars} we  show that \textsc{Max $b$-Save} is polynomial-time solvable for
$k$-caterpillars. Some variants of the \textsc{Max $b$-Save} problem are considered in Section~\ref{s:variants}. Conclusions are given in Section~\ref{s:concl}.

\section{Preliminaries} \label{s:prelim}

All graphs in this paper are undirected, connected, finite and simple. Let $G=(V,E)$ be a graph. An edge in $E$ between vertices $u,v\in V$ will be denoted by $uv$. The \textit{degree} of a vertex $u\in V$, denoted by $deg(v)$, is the number of edges incident to $u$.
We write $G-v$ for the subgraph obtained by deleting a vertex $v$ and all the edges incident to $v$. Similarly, for $A\subseteq V$, we denote by $G-A$ the subgraph of $G$ obtained by deleting the set $A$ and all the edges incident to some vertex in $A$.

\medskip

In order to define the firefighter problem, we use an undirected graph
$G=(V,E)$ and notations of \cite{anshelevich2009}. Each
 vertex in the graph can be in exactly one of the following states:
\textit{burned}, \textit{saved} or \textit{vulnerable}. 
A vertex is said to be burned if it is on fire.
We call a vertex saved if it is either protected by a firefighter --- that is the vertex cannot be burned in subsequent time steps --- or if all paths from any burned vertex to it contains at least one protected vertex.
Any vertex which is neither saved nor burned is called vulnerable.
At time step $t = 0$, all vertices are vulnerable, except vertex $s$,
which is burned. At each time $t
> 0$, at most $b$ vertices can be protected by firefighters and any vulnerable vertex $v$ which is adjacent to a burned
vertex $u$ becomes burned at time $t +
1$, unless it is protected at time step $t$. Burned and saved vertices
remain burned and saved, respectively.

\medskip

%\noindent\textbf{Protection strategy}:
Given a graph $G=(V,E)$ and a vertex $s$  initially on fire, a \textit{protection strategy} is a set $\Phi \subseteq V \times T$ where  $T = \{1,2,\ldots,|V|\}$. We say
that a vertex $v$ is protected at time $t \in T$ according to the
protection strategy $\Phi$ if $(v, t) \in \Phi$. A protection
strategy is \textit{valid} with respect to a budget $b$, if the
following two conditions are satisfied:

\begin{enumerate}
\item if $(v, t) \in \Phi$ then $v$ is not burned at time $t$;
\item let $\Phi_{t}=\{(v,t) \in \Phi\}$; then $|\Phi_{t}| \leq b$ for $t = 1,\ldots,|V|$.
\end{enumerate}

%\noindent\textbf{Propagation rule} :

%At time $t = 0$, all vertices
%are vulnerable, except vertex $s$ which is burned.

Thus at each time $t > 0$, if a vulnerable  vertex $v$ is adjacent to at least one burned vertex and $(v, t) \notin \Phi$, then $v$ gets burned at time $t + 1$.

%\begin{enumerate}
%\item \textit{Non-Spreading Vaccination Model} : At each time $t > 0$, if a
%\textit{vulnerable} vertex $v$ is adjacent to at least one
%\textit{burned} vertex and $(v, t) \notin \Phi$, then $v$ gets
%\textit{burned} at time $t + 1$.
%\item \textit{Spreading Protection Model} : The protection
%is spreading through the graph. At each time step $t > 0$, we
%apply one of two following rules:
%\begin{itemize}
%\item For every \textit{vulnerable} vertex $v$ such that it is adjacent to at least
%one \textit{burned} vertex and $(v, t) \notin \Phi$, then $v$ gets
%\textit{burned} at time $t + 1$.
%\item For every \textit{vulnerable} vertex $v$ such that it is adjacent to at least
%one \textit{burned} vertex, then $v$ gets \textit{protected} at
%time $t + 1$.
%\end{itemize}
%
%\end{enumerate}
%
%\noindent We assume that the protection prevails over the burning
%(see Figure \ref{fig:exfire}).
%
%\begin{figure}[!h]
%\begin{center}
%%\input{exfire.pstex_t}
%\end{center}
%\caption{Illustration of the firefighter problem. Here $b=1$ and
%we use the following strategy: $\Phi = \{(4,1),(6,2),(5,3)\}$.
%Burned vertices are crossed.} \label{fig:exfire}
%\end{figure}
%

\medskip
We define in the following the problems we study.

\medskip
\noindent
 {\sc $b$-Save}

\noindent\textbf{Input}: An undirected graph $G=(V,E)$, a burned vertex $s\in V$, and a subset $S \subseteq V$.

\noindent\textbf{Question}: Is there a valid strategy
$\Phi$ with respect to budget $b$ such that all vertices from $S$ are saved?

\medskip
\noindent 
{\sc Max $b$-Save}

\noindent\textbf{Input}: An undirected graph $G=(V,E)$, a burned vertex $s\in V$, and a subset $S \subseteq V$.

\noindent\textbf{Output}: A valid strategy $\Phi$  with respect to budget $b$ which maximizes the number of saved vertices that belong to $S$.\\

In the figures of the paper, the burned vertices are represented by  black vertices and the vertices from $S$ are represented by $\square$. A protected vertex is represented by $\oplus$.

\medskip

Notice that the NP-hardness of {\sc $b$-Save} implies the NP-hardness of {\sc Max $b$-Save}. Furthermore, if {\sc Max $b$-Save} is solvable in polynomial-time then so is {\sc $b$-Save}.
%For all graph theoretical terms not defined here, the reader is referred to \cite{west}.
%\bigskip
%\noindent\textbf{Graph definitions}
%
%
%A \emph{$k$-star} is a star with legs of length at most $k$ (see
%figure \ref{fig:star}).
%
%
%\begin{figure}[!h]
%\begin{center}
%%\input{star.pstex_t}
%\end{center}
%\caption{A $3$-star.} \label{fig:star}
%\end{figure}
%
% A \emph{$k$-caterpillar} is a caterpillar with legs of length at most
%$k$ (see figure \ref{fig:cater}).
%
%\begin{figure}[!h]
%\begin{center}
%%\input{cater.pstex_t}
%\end{center}
%\caption{A $2$-caterpillar.} \label{fig:cater}
%\end{figure}

\section{Trees} \label{s:trees}

It has been shown in \cite{finbow2007} that {\sc $1$-Save} is NP-complete for trees of maximum degree three using a reduction from not-all-equal 3SAT. Furthermore, {\sc $1$-Save} is polynomial-time solvable for graphs of maximum degree three if the fire breaks out at a vertex of maximum degree two. In this section we generalize these results for any fixed $b\geq 2$.

First of all, we need to define some notions. Let $T$ be a tree and let $s$ be the vertex which is initially burned. From now on, $s$ will be considered as the root of $T$. We define the \textit{level} $k$ of $T$ to be the set of vertices that are at distance exactly $k$ from $s$. The \textit{height} of $T$ is the length of a longest path from $s$ to a leaf. An \textit{ancestor} (resp. \textit{descendant}) of a vertex $v$ in $T$ is any vertex on the path from $s$ to $v$ (resp. from $v$ to a leaf). A \textit{child} of a vertex $v$ in $T$ is an adjacent descendant of $v$. The tree $T$ is said \textit{complete} if every non-leaf vertex has exactly the same number of children.

\begin{remark}
\label{rem:relstrat}
Without loss of generality, we may assume that strategies do not protect a vertex that has a protected ancestor in a tree.
\end{remark}

\begin{remark}
\label{rem:sleave}
For \textsc{$b$-Save} on trees, we may assume without loss of generality that $S$ is the set of leaves. Otherwise for each non-leaf vertex $v \in S$, since we have to save $v$, we can remove the subtree rooted at $v$ such that $v$ becomes a leaf.
\end{remark}

We denote by $\mathcal{T}(r,h,d)$ a complete tree of height $h$ and root $r$ such that every non-leaf vertex has exactly $d$ children and every leaf is at the same distance from the root (see Figure \ref{fig:gadget1}).

\begin{figure}[!h]

\begin{center}
\input{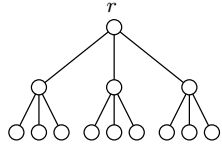}
\end{center}

\caption{$\mathcal{T}(r,2,3)$.}
\label{fig:gadget1}

\end{figure}

For such a tree we obtain the following property.

\begin{lemma}
\label{lem:gadget}
Let $b$ be the number of available firefighters at each time step. Consider a complete tree  $\mathcal{T}(r,h,b+1)$. If the fire breaks out at $r$, then at least one leaf will not be saved.
\end{lemma}

\begin{proof}
Since each non-leaf vertex has exactly $b+1$ children, it follows that at each time step there will be at least one new burning vertex. Thus at the end of the process, at least one leaf will be burned.\qedfill
\end{proof}

\noindent
We also give the following preliminary results.

\begin{lemma}
\label{lem:adj}
Among the strategies that save all the leaves (resp. a maximum number of a given subset of vertices) of a tree, there exists one such that each protected vertex is adjacent to a burning vertex.
\end{lemma}

\begin{proof}
This is a straightforward adaptation of observation 4.1 in \cite{macgillivray2003}.\qedfill
\end{proof}

%\noindent
%Note that the previous lemma does not hold in general graphs (see for instance Figure \ref{fig:awayB}).
%\begin{figure}[!h]
%\begin{center}
%\begin{minipage}{.46\linewidth}
%   \begin{minipage}{.46\linewidth}
%   \begin{center}
%      \input{awayA.pstex_t}
%
%      (a)
%   \end{center}
%   \end{minipage}
%   \begin{minipage}{.46\linewidth}
%   \begin{center}
%      \input{awayB.pstex_t}
%
%      (b)
%   \end{center}
%   \end{minipage}
%\end{minipage}
%\end{center}
%\caption{Sometimes protecting a vertex away from the fire at a given time step is the only way to save $S$.}
%\label{fig:awayB}
%\end{figure}

\begin{lemma}
\label{lem:room}
Let $T$ be a tree and let $\Phi$ be a strategy that saves all the leaves of $T$ using at most $b$ firefighters at each time step. Suppose there exists levels $k$ and $k' > k$ containing $b_{k} \leq b-1$ and $b_{k'} \geq 1$ firefighters, respectively. Then there exists a strategy $\Phi'$ saving all the leaves of $T$ and such that levels $k$ and $k'$ contain $b_{k}+1$ and $b_{k'}-1$ firefighters, respectively.
%Given a tree $T$ and a strategy $\Phi$ that saves all the leaves of $T$, if there exists a level $k$ with $b-1$ firefighters, then moving up a firefighter from a level $k' > k$ to an ancestor at level $k$ leads to a new valid strategy $\Phi'$ that saves at least as much leaves (see Figure \ref{fig:umbrella}).
\end{lemma}

\begin{proof}
Let $v_{k'}$ be a protected vertex by strategy $\Phi$ at a level $k' > k$, and let $v_{k}$ be the ancestor of $v_{k'}$ at level $k$. It follows from Remark \ref{rem:relstrat} that we may assume that $v_{k}$ is not protected. We transform strategy $\Phi$ into a strategy $\Phi'$ as follows (see Figure \ref{fig:umbrella}): protect $v_{k}$ at time step $k$ and do not protect $v_{k'}$ at time step $k'$, that is $\Phi' = (\Phi - \{(v_{k'},k')\}) \cup \{(v_{k},k)\}$. Since $v_{k}$ is an ancestor of $v_{k'}$, it follows that using strategy $\Phi'$, we save a subset of vertices that contains the vertices saved by using $\Phi$. Since level $k$ contains at most $b-1$ firefighters it follows that $\Phi'$ is a valid strategy that saves all the leaves of $T$ and levels $k$ and $k'$ contain respectively $b_{k}+1$ and $b_{k'}-1$ firefighters.
\begin{figure}[!h]

\begin{center}
\input{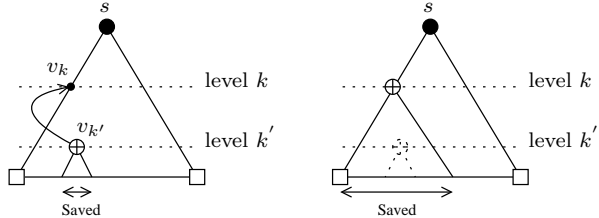}
\end{center}

\caption{Moving up a firefighter leads to a strategy that saves at least the same set of leaves.}
\label{fig:umbrella}

\end{figure}
\qedfill
\end{proof}

%\begin{lemma}
%\label{lem:upb}
%Let $G$ be a graph with $n$ vertices. If $b \geq \frac{n}{2}$ then \textsc{$b$-Save} is polynomial-time solvable on $G$.
%\end{lemma}
%
%\begin{proof}
%If $|N(s) \cap S| > b$ then the answer is \textit{no}. If $|N(s) \cap S| = b$ then the answer is \textit{yes}. So we may assume now that $|N(s) \cap S| < b$. At time step $1$, we protect $b \geq \frac{n}{2}$ vertices including those in $N(s) \cap S$.
%Clearly, we can protect all the remaining unprotected vertices in $S$ at time step $2$. Hence, the answer is \textit{yes}.
%\end{proof}

\noindent
We now give the main result of this section.

\begin{theorem}
\label{th:npc}
For any fixed $b \geq 2$, \textsc{$b$-Save} is NP-complete for
trees of maximum degree $b+2$.
\end{theorem}

\begin{proof}
Clearly, \textsc{$b$-Save} belongs to NP. In order to prove its NP-hardness,
we use a polynomial-time reduction from \textsc{$b$-Save}
for trees of maximum degree $b+2$ to \textsc{$(b+1)$-Save} for
trees of maximum degree $b+3$, for any $b\geq 1$. Since \textsc{$b$-Save} is NP-hard for $b=1$ (see \cite{king2010}), it follows that \textsc{$b$-Save} is NP-complete for any fixed $b \geq 2$.

%To prove the theorem, we show the following: if \textsc{$(s,b,S)$-Save} is NP-complete for trees of maximum degree $b+2$ then \textsc{$(s,b+1,S)$-Save} is NP-complete for trees of maximum degree $b+3$.
Let $I$ be an instance of \textsc{$b$-Save} consisting of a tree
$T=(V,E)$ of maximum degree $b+2$ rooted at some vertex $s$ and a subset $S \subset V$ which corresponds
to the set of leaves. Let $h$ be the height of $T$. We construct an instance $I'$ of
\textsc{$(b+1)$-Save} consisting of a tree $T'=(V', E')$ of maximum degree $b+3$ rooted at some vertex $s'$ and a subset $S' \subset V'$ which corresponds to the leaves of $T'$ as follows (see Figure \ref{fig:ssavereduc1}): add a vertex $s'$; add two paths $\{y_1y_2,\ldots,y_{h-2}y_{h-1}\}$, $\{x_1x_2,\ldots,x_{h-1}x_h\}$, make $y_1,x_1$ adjacent to $s'$ and make $y_{h-1}$ adjacent to $s$; add vertices $v_1,\ldots,v_{b+1}$ and make them adjacent to $s'$; for every vertex $y_i$, $i=1\ldots,h-1$, add vertices $v_{i,1},\ldots,v_{i,b+1}$ and make them adjacent to $y_i$; for $i=1,\ldots,h$ add a path $\{w_{i,1}w_{i,2},\ldots,w_{i,h-1}w_{i,h}\}$ and make $w_{i,1}$ adjacent to $x_i$. This clearly gives us a tree $T'=(V',E')$ of maximum degree $b+3$ rooted at vertex $s'$ and the set of leaves $S'\subset V'$ is given by $S'=S\cup\bigcup_{i=1}^{h-1} \{v_{i,1},\ldots,v_{i,b+1}\}\cup \{w_{1,h},\ldots,w_{h,h}\}\cup\{v_1,\ldots,v_{b+1}\}$.

%add a vertex $s'$ adjacent to $s$; for $i=1,\ldots,h$ add
%a path of length $i$ with endpoints $u_{i}$ and $w_{i}$ such that $u_{i}$ is adjacent to $s'$; finally,
%add $b+1$ vertices $v_{1},\ldots,v_{b+1}$ adjacent to
%$s'$. We set $S' = S \cup \{w_{1},\ldots,w_{h}\} \cup
%\{v_{1},\ldots,v_{b+1}\}$. Notice that the maximum degree of
%$T'$ is $b+2+h$ which is far from being less or equal to $b+3$
%and every leaf is not at the same distance from the root $s'$, but we
%will see later how to fix that.

%\begin{figure}[!h]
%\begin{center}
%\input{ssavereduc1.pstex_t}
%\end{center}
%
%\caption{The construction of $T'$.} \label{fig:ssavereduc1}
%
%\end{figure}

\medskip
We prove now that there exists a strategy $\Phi$ for $I$ that
saves all the vertices in $S$ if and only if there exists a
strategy $\Phi'$ for $I'$ that saves all the vertices in
$S'$.

\medskip
Suppose there exists a strategy $\Phi$ for $I$ that saves all the
vertices in $S$. In order to save all vertices in $S'$, we will
apply strategy $\Phi'$ defined as follows: at time step
$t=1$, we have to protect the vertices $v_{1},\ldots,v_{b+1}$; at each
time step $2\leq t\leq h$, we have to protect the vertices $v_{t-1,1},\ldots,v_{t-1,b+1}$; thus after time step $h$, vertex $s$ is burning as well as vertices $w_{1,h-1},w_{2,h-2},\ldots,$ $w_{h-1,1},x_h$; at each time step $h+1\leq t\leq 2h$, we protect the vertices in $T$ according to $\Phi_{t-h}$ and we use the additional firefighter to protect the leaf $w_{t-h,h}$. This clearly gives us a valid strategy $\Phi'$
saving all the vertices in $S'$.

%protect the vertices in $T$ according to
%$\Phi_{t-1}$ and we use the additional firefighter to
%protect the vertex $w_{t-1}$. 

\medskip
Suppose now that there exists a strategy $\Phi'$ for $I'$ that
saves all the vertices in $S'$. At time step $t=1$, this strategy necessarily consists in protecting vertices $v_1,\ldots,v_{b+1}$. Furthermore, at each time step $2\leq t\leq h$, we have to protect the vertices $v_{t-1,1},\ldots,v_{t-1,b+1}$. It follows from Lemma~\ref{lem:adj} that we may assume that $\Phi'$
is a strategy which, at each time step, protects vertices adjacent to burning vertices.
Thus $\Phi'$ protects, at each time step $k$, at most $b+1$ vertices at level $k$ in $T'$
for $k=h+1,\ldots,2h$.
Let $b_{T}(k)$ be the number of firefighters in the subtree $T$ of $T'$ at level $k$ used by $\Phi'$ and let $\mathcal{B}_{T} = \{k:b_{T}(k) = b+1\}$.
If $\mathcal{B}_{T} = \emptyset$, then for any $k$, $b_{T}(k) \leq b$ and thus the strategy $\Phi'$, restricted to the tree $T$, is a valid strategy for $I$ that saves
all the leaves of $T$. So we may assume now that $\mathcal{B}_{T} \neq \emptyset$.
%We will show that from the strategy $\Phi'$, we can construct a valid strategy for $I$ that saves all the leaves of $T$.

Let $i^{\ell}$ be the $\ell^{th}$ smallest value in $\mathcal{B}_{T}$.
Consider the case $\ell=1$. Suppose that for any $i < i^{1}$,
$b_{T}(i) \geq b$. From the definition of $i^{\ell}$, it follows that we cannot
have $b_{T}(i) = b+1$, thus $b_{T}(i) = b$ for any $ i < i^{1}$. By
construction, this means that, at each time step $i < i^{1}$, the
additional firefighter protects the vertex $w_{i-h,h}$, $i\geq h+1$.
At time step $i^{1}$, since $b_{T}(i^{1}) = b+1$, the vertex
$w_{i^{1}-h,h}$ is not protected and burns which is a contradiction. Thus
there exists a level $i <i_1$ such that $b_{T}(i) < b$.
It follows from Lemma~\ref{lem:room} that there exists a strategy saving the leaves of $T'$
such that $b_{T}(i) \leq b$ and $b_{T}(i^{1}) = b$. Applying this argument
iteratively for $i^{2},\ldots, i^{|\mathcal{B}_{T}|}$, we obtain a strategy $\Phi^{''}$ that saves all the
vertices in $S'$ and such that for any level $k$, $b_{T}(k)
\leq b$. Thus, the strategy $\Phi^{''}$ restricted to the tree
$T$ is a valid strategy that saves all the leaves of $T$.
\end{proof}
%\medskip
%As already mentioned before, our tree $T'$ has not the good properties \textit{i.e.,} maximum degree $b+3$ and
%every leaf is at the same distance from the root. We show now how to fix this.
%We transform $T'$ into a tree $T''$ with the good properties as follows
%(see Figure \ref{fig:ssavereduc2}). Remove the edge $s's$ and the edges
%$s'u_{i}$ for $i=1,\ldots,h$; add two paths $\{x_{1}x_2,\ldots,x_{h-2}x_{h-1}\}$
%and $\{y_{1}y_2,\ldots,y_{h-2}y_{h-1}\}$ such that $x_{1}$ and $y_{1}$ are
%adjacent to $s'$. For $i=1,\ldots,h-2$, add a path of length
%$h-i-2$ and make one endpoint adjacent to $x_{i}$ and the other to
%$u_{i}$; make $u_{h-1}$ and $u_{h}$ adjacent to $x_{h-1}$. For $i=1,\ldots,h-1$, add
%$b+1$ vertices $v_{i,1},\ldots,v_{i,b+1}$ and make them adjacent to $y_{i}$; make $s$
%adjacent to $y_{h-1}$. Finally, we attach a tree $\mathcal{T}(v_{i,j},2h-i-1,b+2)$ to each vertex $v_{i,j}$, $\mathcal{T}(w_{i},h-i,b+2)$ to each vertex $w_{i}$, and $\mathcal{T}(v_{i},2h-1,b+2)$ to each vertex $v_{i}$. One can note that every leaf is at distance $2h$ from the root and the maximum degree is $b+3$ as expected.

\begin{figure}[!h]

\begin{center}
\input{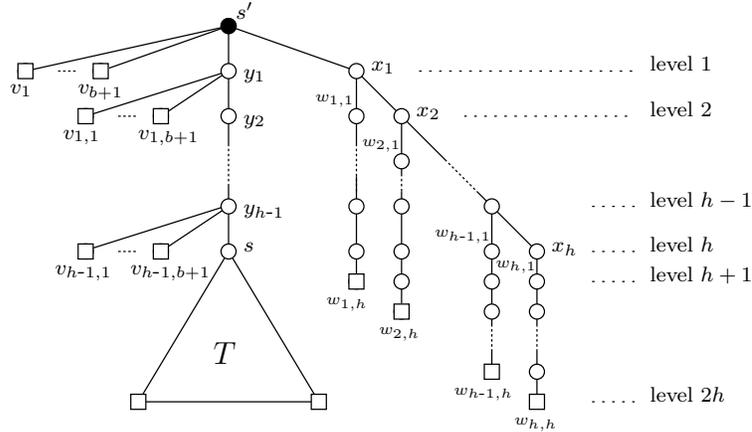}
\end{center}
\caption{The construction of $T'$.}
\label{fig:ssavereduc1}

\end{figure}

Theorem \ref{th:npc} implies that \textsc{Max $b$-Save} is NP-hard for trees of maximum degree $b+2$ when $S$ is the set of leaves.
Notice that Theorem \ref{th:npc} does not imply that \textsc{Max $b$-Save} is NP-hard when $S$ is the set of all vertices. However, the following theorem shows
that this is indeed the case.

 \begin{proposition}
 For any fixed $b \geq 2$, \textsc{Max $b$-Save} is NP-hard for
 trees of maximum degree $b+3$ when $S$ is the set of all vertices.
 \end{proposition}

\begin{proof}
 We construct a polynomial-time reduction from \textsc{$b$-Save} to \textsc{Max $b$-Save} where $b \geq 2$.  
 Let $I$ be an instance of \textsc{$b$-Save} consisting of a tree
 $T=(V,E)$ of maximum degree $b+2$ with $|V| = n$, a burned vertex $s \in V$, and a subset $S \subseteq V$ which corresponds
 to the set of leaves. We construct an instance $I'$ of
 \textsc{Max $b$-Save} consisting of a tree $T'=(V', E')$, a set $S' = V'$, and a positive integer $k$ as follows (see Figure \ref{fig:maxsavereduc}). For every leaf $\ell$ of $T$, add $b+2$ copies $\mathcal{T}_{1,\ell}$, \ldots, $\mathcal{T}_{b+2,\ell}$ of the tree $\mathcal{T}(r, \lceil \mbox{log}_{b+1} n+1 \rceil, b+1)$ such that the root $r_{i,\ell}$ of $\mathcal{T}_{i,\ell}$ is adjacent to $\ell$, for $i \in \{1,\ldots,b+2\}$. Let $|\mathcal{T}|$ denote the cardinality of each of those trees. Notice that each tree $\mathcal{T}_{i,\ell}$ has $|\mathcal{T}| \geq n$ vertices. Set $k = (b+2)|S||\mathcal{T}|$. We will prove that there exists a strategy for $I$ that saves all the vertices in $S$ if and only if there exists a strategy for $I'$ that saves at least $k$ vertices in $S'$.

 Suppose there exists a strategy $\Phi$ for $I$ that saves all the vertices in $S$. Since $S$ is the set of all
 leaves in $T$, it follows that the strategy $\Phi$ applied to $T'$ saves all the vertices of the trees $\mathcal{T}_{i,\ell}$. Notice that we have $(b+2)|S|$ such trees. Thus $\Phi$ saves at least $k=(b+2)|S||\mathcal{T}|$ vertices in $T'$.

Conversely, suppose that no strategy $\Phi$ for $I$ can save all the vertices in $S$. Thus, at least one leaf of $T$ is burned at the end. This necessarily implies that for any strategy $\Phi'$ for $I'$ there is at least one vertex, say $\ell$, of $S$ which is burned. It follows from the construction of $T'$, that in this case there are at least $|\mathcal{T}|$ vertices which will be burned for strategy $\Phi'$. Thus $\Phi'$ saves at most $n-1 + (b+2)|S||\mathcal{T}| - |\mathcal{T}| \leq (b+2)|S||\mathcal{T}| - 1 < k$ vertices.
% Conversely, suppose that any strategy $\Phi$ for $I$ cannot save all the vertices in $S$. Thus, at least one leave $\ell$ of $T$ is burned. By construction, there is at least one tree $\mathcal{T}_{i,\ell}$ with $i \in \{1,\ldots,b+2\}$ such that $\mathcal{T}_{i,\ell}$ is entirely burned. Thus, any strategy for $I'$ saves at most $n + (b+2).|S|.|\mathcal{T}| - |\mathcal{T}| \leq (b+2).|S|.|\mathcal{T}| - 1 < k$.
\end{proof}

 \begin{figure}[!h]
 \begin{center}
 \input{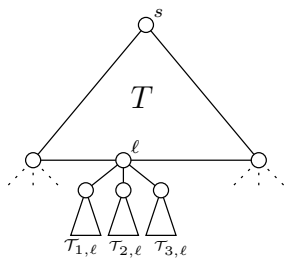}
 \end{center}
 \vspace*{-0.5cm}
 \caption{Construction of $T'$ from the tree $T$ for the case $b=1$.}
 \label{fig:maxsavereduc}
 \end{figure}

The following proposition shows that the sharp separation between the NP-hardness and polynomiality of \textsc{$b$-Save} on trees pointed out in \cite{king2010} is preserved for any fixed $b \geq 2$.

\begin{proposition}
\label{prop:polytree}
Let $b \geq 2$ be any fixed integer and $T$ a tree of maximum degree $b+2$. If the fire breaks out at a vertex
of degree at most $b+1$ then all the leaves of $T$ can be saved if and only if $T$ is not complete. Thus \textsc{$b$-Save} is polynomial-time solvable for trees of maximum degree $b+2$ if the fire breaks out at a vertex of degree at most $b+1$.
\end{proposition}

\begin{proof}
Notice that in this case we protect the vertices such that there is at most one new burning vertex $v$ at each time step. Moreover, the fire stops when the vertex $v$ has degree at most $b+1$.

Suppose that $T$ is not complete. Then there exists a non-leaf vertex $v$ of degree at most $b+1$. From the previous remark we can direct the fire from $s$ to $v$ and stop it. Hence all the leaves of $T$ are saved.

Suppose that $T$ is complete. Then at each time step, there is at least one new burning vertex. Thus there will be
a leave which will burn at the end of the process.

Clearly, verifying whether a tree is complete can be done in polynomial-time.\qedfill
\end{proof}

\begin{remark}
Notice that Proposition \ref{prop:polytree} also holds for \textsc{Max $b$-Save}. Given a subset $S$ of vertices, we direct the fire to a vertex of degree at most $b+1$ such that the number of burned vertices in $S$ is minimum.
\end{remark}

In \cite{finbow2009}, the authors asked whether there exists a constant $c > 1$ such that the degree greedy algorithm that consists, at each time step, to protect a highest degree vertex adjacent to a burning vertex, gives a polynomial-time $c$-approximation for \textsc{Max $1$-Save} for trees. The following proposition answers this question in the case when $b$ firefighters are available at each time step for any $b \geq 1$.

\begin{proposition}
For any $b \geq 1$, there exists no function $f:N \to (1,+\infty)$ such that the degree greedy algorithm is an $f(n)$-approximation algorithm for \textsc{Max $b$-Save} for trees where $S$ is the set of all vertices. 
\end{proposition}

\begin{proof}
Consider a tree $\mathcal{T}(r, h-1, b+1)$ where $h$ is a positive integer. Add a vertex $s$ adjacent to $r$ and a vertex $v_1$ adjacent to $s$; for $i = 2, \ldots, h-1$, add a path of length $i-1$ with endpoints $u_i$ and $v_i$ such that $u_i$ is adjacent to $s$; finally, for $i = 1, \ldots, h-1$, add $b+2$ vertices adjacent to $v_i$ (see Figure \ref{fig:greedegree}).

Notice that the degree greedy algorithm protects vertices in the following order: $v_1, \ldots, v_{h-1}$. Thus it saves $g_h = (h-1)(b+2)$ vertices. However, it is not difficult to see that the optimal solution protects vertices  in the following order: $r, v_2, \ldots, v_{h-1}$. Thus, in an optimal solution we save $opt_h = (h-2)(b+2) + \sum_{i=0}^{h-1} (b+1)^{i}$ vertices. Since $\frac{opt_h}{g_h} \rightarrow +\infty$ when $h \rightarrow +\infty$, the result follows.
\end{proof}

\begin{figure}[!h]
\begin{center}
\input{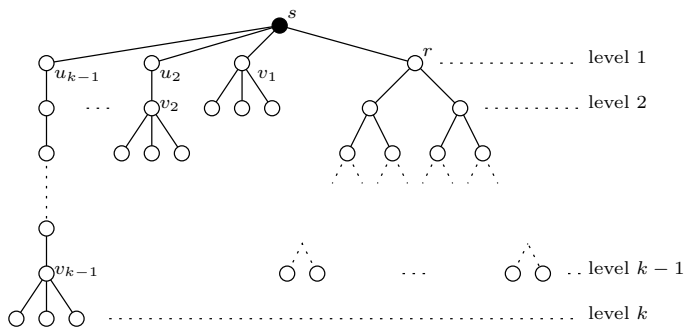}
\end{center}
\caption{Instance where the degree greedy algorithm gives no approximation guarantee for the case $b=1$. Since here $S=V$, we did not represent vertices in $S$ by squares.}
\label{fig:greedegree}
\end{figure}

\section{$k$-caterpillars}\label{s:caterpillars}
 
In this section, we will present a subclass of trees for which \probtitle{Max $b$-Save} is polynomial-time solvable for $b \geq 1$.

A \textit{caterpillar} is a tree $T$ such that the vertices of $T$ with degree at least $2$ induce a path. In other words, a caterpillar $T$ consists of a path $P$ such that all edges have at least one endpoint in $P$. A \textit{$k$-caterpillar}, $k\geq 1$, is a caterpillar in which any pending edge $uv$, with $u\in V(P)$, $v\not \in V(P)$ (\textit{i.e.,} any edge with exactly one endpoint in $P$) may be replaced by a path of length at most $k$ (see Figure \ref{fig:caterstar}). This path is then called a \textit{leg} of the $k$-caterpillar at vertex $u$. Thus a caterpillar is a $1$-caterpillar.

A \textit{star} is a tree consisting of one vertex, called the \textit{center} of the star, adjacent to all the others. Thus a star on $n$ vertices is isomorphic to the complete bipartite graph $K_{1,n-1}$. A \textit{$k$-star}, $k\geq 1$, is a tree obtained from a star in which any edge may be replaced by a path of length at most $k$ (see Figure \ref{fig:caterstar}). Thus a star is a $1$-star.
Notice that a $k$-star is a special case of a $k$-caterpillar

\begin{figure}[!h]
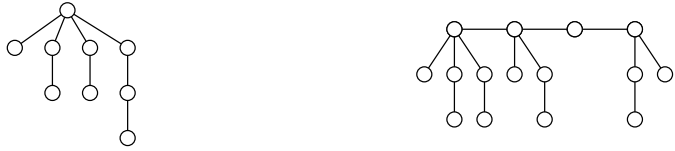

\begin{center}
\begin{minipage}{.90\linewidth}
   \begin{minipage}{.46\linewidth}
   \begin{center}
      \input{star.pstex_t}
   \end{center}
   \end{minipage}
   \begin{minipage}{.46\linewidth}
   \begin{center}
      \input{cater.pstex_t}
   \end{center}
   \end{minipage}
\end{minipage}
\end{center}

\caption{A $3$-star (left) and a $2$-caterpillar (right).}
\label{fig:caterstar}

\end{figure}

In \cite{macgillivray2003}, the authors showed that the degree greedy algorithm gives an optimal solution for \probtitle{Max $1$-Save} on caterpillars when $S=V$. However, this result does not hold
for $k$-caterpillars, see for instance Figure \ref{fig:caternopt}.

\begin{figure}[!h]
\begin{center}
\input{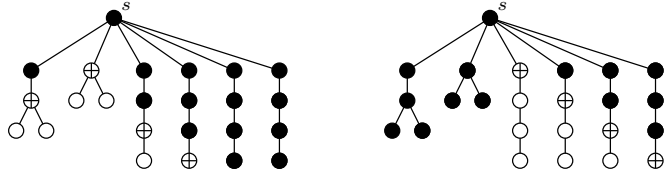}
\end{center}

\caption{A $k$-caterpillar for which the degree greedy algorithm (left) does not give the optimal solution (right). Since here $S=V$, we did not represent vertices in $S$ by squares.}
\label{fig:caternopt}

\end{figure}

In this section,
we give a polynomial-time algorithm for \probtitle{Max $b$-Save} for $k$-caterpillars for any $b \geq 1$ and $S \subseteq V$. In order to prove our main result of this section we first need to show the following.

\begin{theorem}
\label{th:maxbssave_kstar}
For any $k\geq 1$, $b\geq 1$, \probtitle{Max $b$-Save} is polynomial-time solvable for $k$-stars.
\end{theorem}

\begin{proof}
We construct a polynomial-time reduction from \probtitle{Max $b$-Save} to the \probtitle{Min Cost Flow} problem which is known to be polynomial-time solvable (see for instance \cite{orlin}).
Let $G=(V,E)$ be a $k$-star. First assume that $s\in V$ is the center of $G$. Let $d=deg(s)$. Let $P_1=\{sv_{11},v_{11}v_{21},\ldots,$ $v_{(p_1-1)1}v_{p_11}\}$, $\ldots,$ $P_{d}=\{sv_{1d},v_{1d}v_{2d},\ldots,v_{(p_d-1)d}v_{p_dd}\}$ be the maximal paths of $G$ starting at vertex $s$, with $p_1,\ldots,p_d\leq k$ and $v_{0j}=s$, for $j=1,\ldots,d$, if it exists. Let $p=\max\{p_1,\ldots,p_d\}$. For each vertex $v_{ij}$ in these paths, we define $S_{ij}=\{v_{ij},v_{(i+1)j},\ldots v_{p_jj}\}\cap S$. Notice that we may assume that every path $P_j$ contains at least one vertex of $S$ (otherwise we may delete $V(P_j)\setminus\{s\}$).

We construct an auxiliary digraph $G'=(V',U')$ (see Figure \ref{fig:mincostflow}), where $V'=\{L_1,\ldots,L_p\}\cup \{C_1,\ldots,C_d\}$ $\cup \{\ell,r\}$ and $U'=\{(L_i,C_j)|\ v_{ij}\in P_j\}\cup \{(\ell,L_i)|\ i=1,\ldots,p\}\cup \{(C_j,r)|\ j=1,\ldots,d\}$. In this digraph $G'$, we associate with each arc $(L_i,C_j)$, a cost $u(i,j)=-|S_{ij}|$. All other arcs have cost zero. Furthermore we associate with each arc $(\ell,L_i)$ a capacity $c(\ell,i)=b$, with each arc $(L_i,C_j)$ a capacity $c(i,j)=1$ and with each arc $(C_j,r)$ a capacity $c(j,r)=1$. Finally we associate a supply of value $d$ with vertex $\ell$ and a demand of value $-d$ with vertex $r$ (all other vertices have a supply and a demand equal to zero). Thus we obtain an instance of \probtitle{Min Cost Flow} (we want to satisfy the supply and demand of each vertex with a minimum total cost and such that the capacity constraints are respected) and clearly $G'$ can be obtained from $G$ in polynomial-time.\\

\begin{figure}[!h]
\begin{center}
\input{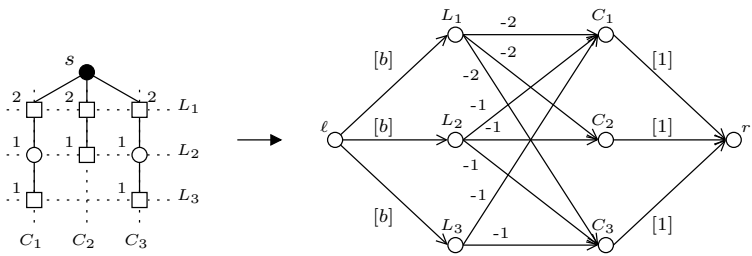}
\end{center}

\caption{The auxiliary digraph $G'$.}
\label{fig:mincostflow}

\end{figure}

We show now that solving \probtitle{Max $b$-Save} in $G$ is equivalent to solving \probtitle{Min Cost Flow} in $G'$. Consider a feasible solution of \probtitle{Max $b$-Save} in $G$ of value $\nu$. We may assume without loss of generality (see Remark \ref{rem:relstrat}) that at most one vertex is protected in each path $P_j$, $j\in \{1,\ldots,d\}$, and (see Lemma \ref{lem:adj}) that at most $b$ vertices are protected in each set $V_i=\{v_{i1},v_{i,2},\ldots,v_{id}\}$, $i\in \{1,\ldots,p\}$ (notice that some of these vertices $v_{ij}, j=1,\ldots,d$, may not exist in $G$). Let $\mathcal{D}=\{v_{ij}|\ v_{ij}\ \mbox{is protected},\ i\in \{1,\ldots,p\}, j\in \{1,\ldots,d\}\}$. Thus $\nu=\sum_{v_{ij}\in \mathcal{D}}|S_{ij}|$. Consider now some vertex $v_{ij}\in \mathcal{D}$. Then in $G'$, we will use one flow unit on the path $\ell$-$L_i$-$C_j$-$r$. Repeating this procedure for every vertex in $\mathcal{D}$, we obtain a flow in $G'$ of value $|\mathcal{D}|$ and of cost $\sum_{v_{ij}\in\mathcal{D}}u(i,j)=\sum_{v_{ij}\in\mathcal{D}}-|S_{ij}|=-\nu$. Since at most $b$ vertices are protected in each set $V_i$, it follows that at most $b$ units of flow use the arc $(\ell,L_i)$, for $i\in \{1,\ldots,p\}$. Furthermore, since exactly one vertex is protected in each path $P_j$, it follows that exactly one flow unit uses the arc $(C_j,r)$ for $j\in\{1,\ldots,d\}$. Hence, we obtain a feasible solution of \probtitle{Min Cost Flow} in $G'$.

Conversely, consider now a feasible solution of \probtitle{Min Cost Flow} in $G'$ of value $-\mu$. Let $\mathcal{A}$ be the set of arcs $(L_i,C_j)$ used by a flow unit, $i\in \{1,\ldots,p\}$, $j\in \{1,\ldots,d\}$. Thus $-\mu=\sum_{(L_i,C_j)\in \mathcal{A}}-|S_{ij}|$. For each flow unit on a path $\ell$-$L_i$-$C_j$-$r$, we choose vertex $v_{ij}$ in $G$ to be protected, for $i\in \{1,\ldots,p\}$, $j\in \{1,\ldots,d\}$. Since the capacity of an arc $(\ell,L_i)$ is $b$, at most $b$ vertices in $V_i$ will be chosen to be protected, $i\in \{1,\ldots,p\}$. Let us denote by $V_{i}^{*}$ the set of vertices in $V_i$ chosen to be protected. Furthermore, since the capacity of an arc $(C_j,r)$ is one, exactly one vertex in each path $P_j$ will be chosen to be protected, $j\in \{1,\ldots,d\}$. Thus, if we protect at each time step $i$ the vertices in $V_{i}^{*}$, we obtain a feasible solution of \probtitle{Max $b$-Save} in $G$ of value $\sum_{i}\sum_{v_{ij}\in V_{i}^{*}}|S_{ij}|=\mu$.\\

Finally, we have to consider the case when $s$ is not the center of $G$. The case when $s$ has degree one is trivial. Thus we may assume now that $deg(s)=2$. If $b\geq 2$, we are done. Thus we may assume now that $b=1$. If both neighbors of $s$ are in $S$, then the optimal solution is clearly $|S|-1$. If both neighbors of $s$ are not in $S$, then the optimal solution is clearly $|S|$. Hence the only case remaining is when exactly one neighbor of $s$ is in $S$. Let $u_1,u_2$ be the neighbors of $s$ such that $u_1\in S,u_2\not\in S$. If $u_2$ is not the center of $G$, the optimal solution is clearly $|S|$. Thus we may assume now that $u_2$ is the center of $G$. Let $Q$ denote the set of vertices of the unique maximal path starting at vertex $u_2$ and containg $u_1$. In that case we have to compare the value of two solutions: (i) $|S|-1$ which is the value of the solution obtained by protecting first $u_2$ and then, during the second time step, we protect the neighbor of $u_1$ which is not $s$ (if it exists); (ii) the value of the solution obtained by protecting first $u_1$ and then applying our algorithm described above to the graph $G-(Q\setminus \{u_2\})$ (\textit{i.e.,} by reducing our problem to a \probtitle{Min Cost Flow} problem).\qedfill
\end{proof}

\begin{remark}
\label{kstar_general}
Notice that the polynomial reduction from \probtitle{Max $b$-Save} to \probtitle{Min Cost Flow} described in the proof of Theorem \ref{th:maxbssave_kstar} is still valid if the number of vertices that can be protected at each time step is not constant (for instance if we are allowed to protect at most $b_1$ vertices during the first time step, $b_2$ vertices during the second time step, etc...). In that case we just need to adapt the capacity of the arcs $(\ell,L_i)$ accordingly.

Furthermore the polynomial reduction remains valid in the case where some of the vertices in a set $V_i$ are not allowed to be protected during time step $i$. In this case we simply do not put an arc from $L_i$ to the corresponding vertices $C_j$ in $G'$.
\end{remark}

Consider now a $k$-caterpillar $G=(V,E)$. Let $P$ be the path in the caterpillar from which $G$ has been obtained, which is induced by vertices of degree at least two. We will call $P$ the \textit{spine} of the $k$-caterpillar.

We are now ready to prove the main result of this section.

\begin{theorem}
\label{th:maxbssave_kcater}
For any $k \geq 1$, $b \geq 1$, \probtitle{Max $b$-Save} is polynomial-time solvable for $k$-caterpillars.
\end{theorem}

\begin{proof}
Let $G=(V,E)$ be a $k$-caterpillar and let $P=\{v_1v_2,v_2v_3,\ldots,v_{p-1}v_p\}$ be the spine of $G$. First assume that $s$ is a vertex of $P$, say $s=v_i$, $i\in \{1,\ldots,p\}$. Let $P_1=\{v_1v_2,\ldots,v_{i-2}v_{i-1}\}$ and $P_2=\{v_{i+1}v_{i+2},\ldots,v_{p-1}v_p\}$. It follows from Remark \ref{rem:relstrat} that we may assume that at most one vertex is protected in $P_1$ and at most one vertex is protected in $P_2$. Consider a strategy in which we decide to protect exactly two vertices of $P$, say vertex $v_j$, for $j\in\{1,\ldots,i-1\}$ and vertex $v_q$, for $q\in \{i+1,\ldots,p\}$. We may assume that $v_j$ is protected during time step $i-j$ and vertex $v_q$ is protected during time step $q-i$ (see Lemma \ref{lem:adj}). Notice that the vertices $v_{j+1},\ldots,v_{i-1},v_{i+1},\ldots,v_{q-1}$ will not be protected in this strategy. Construct a $(k+p)$-star $G'$ as follows (see Figure \ref{fig:catertostar}):

\begin{itemize}
\item[(a)] delete all vertices $v_1,\ldots,v_j$ as well as the legs at these vertices (all these vertices are saved in our strategy);
\item[(b)] delete all vertices $v_q,\ldots,v_p$ as well as the legs at these vertices (all these vertices are saved in our strategy);
\item[(c)] delete all edges of $P$;
\item[(d)] for every $r\in \{j+1,\ldots,i-1,i+1,\ldots,q-1\}$, let $u_1^r,\ldots,u_{d(v_r)-2}^r$ be the neighbors of $v_r$ not belonging to $P$; delete $v_r$ and replace it by $d(v_r)-2$ vertices $v_{1}^{r},\ldots,v_{d(v_r)-2}^{r}$ such that $v_{l}^{r}$ is adjacent to $u_l^r$ for $l\in \{1,\ldots,d(v_r)-2\}$;
\item[(e)] join every vertex $v_{\ell}^{r}$, for $r\in \{j+1,\ldots,i-1\}$ and $\ell\in \{1,\ldots,d(v_r)-2\}$, to $v_i$ by a path $P^{r\ell}$ of length $i-r$;
\item[(f)] join every vertex $v_{\ell}^{r}$, for $r\in \{i+1,\ldots,q-1\}$ and $\ell\in \{1,\ldots,d(v_r)-2\}$, to $v_i$ by a path $P^{r\ell}$ of length $r-i$;
\end{itemize}

\begin{figure}[!h]
\begin{center}

\input{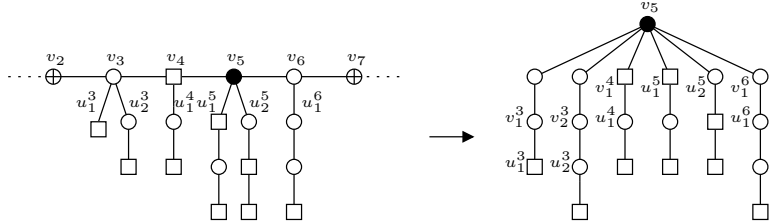}
\end{center}

\caption{The construction of $G'$ with $i=5$, $j=2$, and $q=7$.}
\label{fig:catertostar}

\end{figure}

From the above construction it follows that $G'$ is a $(k+p)$-star with center $v_i$. Now in order to solve our initial problem, we need to solve  \probtitle{Max $b$-Save} in $G'$ with the following additional constraints: for every $r\in \{j+1,\ldots,i-1,i+1,\ldots,q-1\}$ and every $\ell\in \{1,\ldots,d(v_r)-2\}$ we are not allowed to protect the vertices of $V(P^{r\ell})$. Indeed, since we decided to protect $v_j$ and $v_q$, the vertices $v_{j+1},\ldots,v_{i-1},v_{i+1},\ldots,v_{q-1}$ will not be saved. Notice that these vertices are represented by the vertices of paths $P^{r\ell}$ in $G'$. Moreover, if $i-j \neq q-j$ then at time steps $i-j$ and $q-j$ only $b-1$ firefighters are available (since we protect $v_j$ and $v_q$ at these time steps); if $i-j = q-j$ then only $b-2$ firefighters are available at time step $i-j$. It follows from Theorem \ref{th:maxbssave_kstar} and Remark \ref{kstar_general} that this problem can be solved in polynomial-time.

Since the number of choices of a pair of vertices $(v_j,v_q)$ to be protected on $P$ is $(i-1)\times (p-i)$, we can determine in polynomial-time the best strategy to adopt if we want to protect exactly two vertices on $P$. Notice that a similar procedure to the one described above can be used if we decide to protect exactly one vertex on $P$ respectively if we decide not to protect any vertex of $P$. Clearly the number of choices of exactly one vertex $v_j$, $j\in \{1,\ldots,i-1,i+1,\ldots,p\}$, to be protected on $P$ is $p-1$. Thus we conclude that if $s\in V(P)$ we can determine an optimal strategy in polynomial-time.\\

It remains the case when $s\not\in V(P)$. Similar to the proof of Theorem \ref{th:maxbssave_kstar}, we will distinguish several cases. The case when $s$ has degree one is trivial. Thus we may assume now that $deg(s)=2$. If $b\geq 2$, we are done. Thus we may assume now that $b=1$. If both neighbors of $s$ are in $S$, then the optimal solution is clearly $|S|-1$. If both neighbors of $s$ are not in $S$, then the optimal solution is clearly $|S|$. Hence the only case remaining is when exactly one neighbor of $s$ is in $S$. Let $u_1,u_2$ be the neighbors of $s$ such that $u_1\in S,u_2\not\in S$. If $u_2\not\in V(P)$, the optimal solution is clearly $|S|$. Thus we may assume now that $u_2\in V(P)$. In this case we have to compare the value of two solutions: (i) $|S|-1$ which is the value of the solution obtained by protecting first $u_2$ and then, during the second time step, we protect the neighbor of $u_1$ which is not $s$ (if it exists); (ii) the value of the solution obtained by protecting first $u_1$ and then applying our algorithm described above to the graph $G-(Q\setminus \{u_2\})$, where $Q$ is the set of vertices of the unique maximal path starting at $u_2$ and containing $u_1$.
\qedfill
\end{proof}

\section{Variants of {\sc Max $b$-Save}} \label{s:variants}

In this section, we give some results for a weighted version of {\sc Max $b$-Save} as well as for its complementary version. 

\subsection{Weighted version}

We would like to mention that our positive results (Proposition~\ref{prop:polytree}, Theorem~\ref{th:maxbssave_kstar}, and Theorem~\ref{th:maxbssave_kcater}) may be generalized to a weighted version of \probtitle{Max $b$-Save}.

Suppose that we are given a weight $w(v)$ for each vertex $v \in S \subseteq V$. These weights may for instance reflect the importance of the vertices: if $w(v_1) > w(v_2)$, vertex $v_1$ is considered as more important than vertex $v_2$.
Then we may define the following problem:

\medskip

\noindent 
{\sc Max Weighted $b$-Save}

\noindent\textbf{Input}: An undirected graph $G=(V,E)$, a burned vertex $s\in V$, a subset $S \subseteq V$, and a weight function $w:S \to N$.

\noindent\textbf{Output}: A valid strategy $\Phi$  with respect to budget $b$ which maximizes the total weight of the saved vertices that belong to $S$.\\

\medskip

In the proof of Proposition~\ref{prop:polytree}, if we direct the fire to a vertex of degree at most $b+1$ such that the total weight of the burned vertices in $S$ is minimum then we get the following result.

\begin{proposition}
For any $b\geq 1$, \probtitle{Max Weighted $b$-Save} is polynomial-time solvable for trees of maximum degree $b+2$ if the fire breaks out at a vertex of degree at most $b+1$.
\end{proposition}

Now by replacing the costs $u(i,j)$ in the proof of Theorem~\ref{th:maxbssave_kstar} by $u(i,j) = - |\sum_{v \in S_{ij}} w(v)|$ and adapting the case when $s$ is not the center of $G$ according to the weights, it is not difficult to see that we obtain the following.

\begin{theorem}
For any $k\geq 1$, $b\geq 1$, \probtitle{Max Weighted $b$-Save} is polynomial-time solvable for $k$-stars.
\end{theorem}

Using this result and adapting the case when $s \notin V(P)$ according to the weights, it is straightforward that we obtain the following result.

\begin{theorem}
For any $k\geq 1$, $b\geq 1$, \probtitle{Max Weighted $b$-Save} is polynomial-time solvable for $k$-caterpillars.
\end{theorem}

Although the results above are more general than the results in Sections \ref{s:trees} and \ref{s:caterpillars}, we decided to present in detail the results concerning \probtitle{Max $b$-Save} in this paper, since this corresponds to the version which has been widely studied in the literature.

\subsection{Min version}

Let us consider now the minimum version of the {\sc Max $b$-Save} problem which is defined as follows. 

\medskip

\noindent 
{\sc Min $b$-Save}

\noindent\textbf{Input}: An undirected graph $G=(V,E)$, a burned vertex $s\in V$, a subset $S \subseteq V$.

\noindent\textbf{Output}: A valid strategy $\Phi$ with respect to budget $b$ which minimizes the number of burned vertices that belong to $S$.\\

In contrast to {\sc Max $b$-Save} which is constant approximable on trees, the following theorem shows a strong inapproximability result for {\sc Min $b$-Save} even when restricted to trees.

\begin{theorem}
For any $\epsilon \in (0,1)$ and any $b\geq 1$, \probtitle{Min $b$-Save} is not $n^{1-\epsilon}$-approximable 
even for trees on $n$ vertices when $S$ is the set of all vertices, unless $P=NP$.
\end{theorem}

\begin{proof}
We construct a polynomial-time reduction from \probtitle{$b$-Save} to \probtitle{Min $b$-Save}.
Let $I$ be an instance of \textsc{$b$-Save} consisting of a tree
$T=(V,E)$ with $|V| = n_1$, a burned vertex $s \in V$, and a subset $S \subseteq V$ which corresponds
to the set of leaves. We construct an instance $I'$ of
\textsc{Min $b$-Save} consisting of a tree $T'=(V', E')$ with $|V'|=n$, a burned vertex $s'$, and $S' = V'$ as follows. For every leaf $\ell$ of $T$, add $\lfloor n_1^{\beta}+b \rfloor$ vertices adjacent to $\ell$ where $\beta = \frac{4}{\epsilon}-3$. Notice that $n = \lfloor n_1^{\beta}+b \rfloor |S| + n_1 < n_1^{\beta + 3}$.

If there exists a strategy that saves all the vertices in $S$ then at most $n_1$ vertices are burned in $V'$. Conversely, if there is no strategy that saves all the vertices in $S$ then at least $n_1^{\beta}$ vertices are burned in $V'$. %Thus there is a gap of 
%there is no approximaton algorithm with guarantee ${n'}^{1-\epsilon}$ unless $P=NP$.

Suppose that there exists a polynomial-time $n^{1-\epsilon}$-approximation algorithm $A$ for \probtitle{Min $b$-Save}. Thus, if $I$ is a \textit{yes}-instance, the algorithm gives a solution of value $A(I') \leq n^{1-\epsilon}n_1 < n_1^{(\beta + 3)(1-\epsilon)+1} = n_1^{\beta}$.
%\centerline{$v \geq \frac{n^{\beta}}{{n'}^{1-\epsilon}} = \frac{n^{\beta}}{(n^{\beta}+n)^{1-\epsilon}} \geq \frac{n^{\beta}}{(2n^{\beta})^{1-\epsilon}} = \frac{n^{2}}{2^{1-\epsilon}} \geq \frac{n^{2}}{2}$}
%\noindent
If $I$ is a \textit{no}-instance, the solution value is $A(I') \geq n_1^{\beta}$. Hence, the approximation algorithm $A$ can distinguish in polynomial time between \textit{yes}-instances and \textit{no}-instances for \probtitle{$b$-Save} implying that $P=NP$.
\end{proof}

\section{Conclusion}\label{s:concl}

In this paper, we studied some generalizations and variants of the firefighter problem when more than one firefighter is available at each time step and we answered three open questions of \cite{finbow2009}.
%The firefighter problem was very well studied as indicated in the survey \cite{finbow2009}. In this paper, we mainly answer three open problems from this survey. First, we study the 
%complexity of the firefighter problem on trees when $b$ vertices can be protected
% by firefighters at each time step with $b\geq 2$: the problem is NP-hard for trees with maximum degree $b+2$
%   but polynomial-time solvable for trees with maximum degree $b+2$ when the fire breaks out at a vertex of
%degree at most $b+1$. Second, we show that the degree greedy strategy gives no approximation guarantee for the firefighter problem when $b \geq 1$ firefighters are available at each time step.
%Third,we identify a subclass of trees for which
% the firefighter problem with $b$ firefighters at each time step, $b\geq 1$, can be solved in polynomial-time, namely $k$-caterpillars.
 Several interesting questions remain open. The complexity of \probtitle{$b$-Save} and \probtitle{Max $b$-Save} in the following cases are not known: when the number of firefighters at each time step depends on the number of vertices; when every leaf is at the same level. The complexity of \probtitle{Max $b$-Save} for trees of maximum degree $b+2$ is not establish.
% for trees of maximum degree $b+2$ when $S$ is the set of all vertices and every leaf is at the same level; when the number of firefighters at each time step depends on the number of vertices. The complexity of \probtitle{$b$-Save} for trees of maximum degree $b+2$ when every leaf is at the same level is not establish.  
 Finally, the problem is 2-approximable for trees when $S$ is the set of vertices. Establishing non approximability results or better approximability results is another open problem.

%%%-----------------------------------------------------------------------

%\bibliographystyle{model1a-num-names}
\bibliographystyle{abbrv}
\bibliography{biblio}
%\nocite{*}

\end{document}